\documentclass{article}
\date{
April 13, 2014%
}
\usepackage[english]{babel}
\usepackage{graphicx}
\usepackage[a4paper]{geometry}
\parskip=\smallskipamount
\advance\textwidth by -2cm
\advance\oddsidemargin by 1cm
\advance\evensidemargin by 1cm

\usepackage{bm,amsmath,amssymb}

\newcommand{\signum}{\operatorname{sgn}}
\newcommand{\Diag}[1]{\operatorname{Diag}(#1)}
\newcommand{\w}{{\bm w}}
\newcommand{\x}{{\bm x}}
\newcommand{\Gauss}{Gau\ss}
\newcommand{\1}{\mathbf 1}
\renewcommand{\epsilon}{\varepsilon}
\renewcommand{\emptyset}{\varnothing}

\newtheorem{definition}{Definition}
\newtheorem{proposition}{Proposition}
\newtheorem{theorem}{Theorem}
\newtheorem{corollary}{Corollary}

\newcounter{noqed}
\newcommand{\qed}{ \ifmmode\mbox{ }\fi\rule[-.05em]{.3em}{.7em}\setcounter{noqed}{0}}
\newenvironment{proof}[1][{}]{\noindent{\bf Proof#1. }\setcounter{noqed}{1}}{\ifnum\value{noqed}=1\qed\fi\par\medskip}

\title{Supremum--Norm Convergence for Step--Asynchronous Successive
Overrelaxation on M-matrices}
\author{Sebastiano Vigna\thanks{The author was supported by the EU-FET
grant NADINE (GA 288956).}}
\begin{document}
\bibliographystyle{alpha}

\maketitle

\begin{abstract}
Step-asynchronous successive overrelaxation updates the values contained in
a single vector using the usual \Gauss--Seidel-like weighted rule, but 
arbitrarily mixing old and new values, the only constraint being temporal
coherence---you cannot use a value before it has been
computed. We show that given a nonnegative real matrix
$A$, a $\sigma\geq\rho(A)$ and a vector $\w>0$ such that $A\w\leq\sigma\w$, every iteration of
step-asynchronous successive overrelaxation for the problem $(sI- A)\x=\bm
b$, with $s >\sigma$, reduces geometrically the $\w$-norm of the current error
by a factor that we can compute explicitly. Then, we show that given a
$\sigma>\rho(A)$ it is in principle always possible to compute such a $\w$.
This property makes it possible to estimate the supremum norm of the absolute
error at each iteration without any additional hypothesis on $A$, even when $A$
is so large that computing the product $A\bm x$ is feasible, but estimating the
supremum norm of $(sI-A)^{-1}$ is not.
\end{abstract}

methods for linear systems)


\section{Introduction}

We are interested in providing \emph{computable absolute bounds in $\ell_\infty$ norm} 
on the convergence of
a mildly asynchronous version of successive overrelaxation (SOR)
applied to problems of the form $(s  I- A)\bm x = \bm b$, where $A$ is a
nonnegative real matrix and $s >\rho(A)$. A matrix of the form $s  I- A$ under
these hypotheses is called a \emph{nonsingular M-matrix}~\cite{BePNMMS}.

We stress from the start that there are no other hypotheses on $A$ such as
irreducibility, symmetry, positive definiteness or (weak) 2-cyclicity, and that
$A$ is assumed to be very large---so large that computing $A\bm x$ (or
performing a SOR iteration) is feasible (maybe streaming over the matrix
entries), but estimating $\bigl\|(sI-A)^{-1}\bigr\|_\infty$ is not.

Our main motivation is the parallel computation with arbitrary guaranteed
precision of various kinds of \emph{spectral rankings with
damping}~\cite{VigSR}, most notably Katz's index~\cite{KatNSIDSA} and
PageRank~\cite{PBMPCR}, which are solutions of problems of the form above with
$A$ derived from the adjacency matrix of a very large graph, the only relevant
difference being that the rows of $A$ are $\ell_1$-normalized in the case of PageRank.

By ``computable'' we mean that there must be a finite computational process that
provides a bound on $\bigl\|\bar{\bm x}-\bm x^{(t)}\bigr\|_\infty$,
where $\bar{\bm x}$ is the solution and $\bm x^{(t)}$ is the $t$-th
approximation. Such a bound would make it possible to claim that we
know the solution up to some given number of significant fractional digits. For
example, without further assumptions on $A$ convergence results based on the
spectral radius are not computable in this sense and results concerning the
residual are not applicable because of the unfeasibility of estimating
$\bigl\|(sI-A)^{-1}\bigr\|_\infty$. 

We are also interested in highly parallel versions for modern multicore systems.
While SOR and other iterative methods are apparently strictly sequential
algorithms, there is a large body of literature that studies what happens when
updates are executed in arbitrary order, mixing old and new values. Essentially,
as long as old values come from a finite time horizon (e.g., there is a finite
bound on the ``oldness'' of a value) convergence has been proved for all major
standard sequential hypothesis of convergence\footnote{It is a bit surprising,
indeed, that the statement that \Gauss--Seidel is difficult to parallelize
appears so often in the literature. In a sense, an algorithm updating in
arbitrary order using possibly old values is not any longer \Gauss--Seidel. On
the other hand, this is exactly what one expects when asking the question ``is
\Gauss--Seidel parallelizable''?} (for the main results, see the sections
about \emph{partial asynchrony} in Bertsekas and Tsitsiklis's encyclopedic book~\cite{BeTPDCNM}).

Again, however, results are always stated in terms of convergence in the
limit, and the speed of convergence, which decays as the time horizon gets
larger, often cannot be stated explicitly. Moreover, the theory is modeled
around message-passing systems, where processor might actually use very old
values due to transmission delays. In the multicore, shared-memory system
application we have in mind it is reasonable to assume that after each iteration
memory is synchronized and all processors have the same view.

Our main motivation is obtaining (almost) ``noise-free'' scores to perform
accurate comparisons of the induced rankings using Kendall's $\tau$~\cite{KenTTRP}:
\[\tau(\bm r, \bm s)
:=\frac{\sum_{i<j}\signum(r_i-r_j)\signum(s_i-s_j)}{\sqrt{\sum_{i<j}\signum(r_i-r_j)^2}
\sqrt{\sum_{i<j}\signum(s_i-s_j)^2}}. \] 
Computational noise can be quite problematic in evaluating Kendall's $\tau$ because
the signum function has no way to distinguish large and small differences---they are all mapped to $1$ or $-1$~\cite{BPSTPTP}.

Suppose, for example, that we have a
graph with a large number $n$ of nodes, and some centrality index that assigns score $0$ the first
$n/2$ nodes and score $1$ the remaining nodes. Suppose we have also another
index assigning the same scores, and that this new index is defined
by an iterative process, which is stopped at some point (e.g., an iterative
solver for linear systems). If the computed values
include computational random noise and evaluate $\tau$ on the two
vectors, we will obtain a $\tau$ close to $1/\sqrt2\approx 0.707$, even if the ranks are
perfectly correlated. On the other hand, with a sufficiently small guaranteed
absolute error we can proceed to truncate or round the second set of scores, obtaining
a result closer to the real correlation.

This scenario is not artificial: when comparing, for instance, indegree
with an index computed iteratively (e.g., Katz's index, PageRank, etc.), we have
a similar situation. Surprisingly, the noise from iterative computations can
even \emph{increase} correlation (e.g., between the dominant eigenvector of a graph
that is not strongly connected and Katz's index, as the residual score in nodes
whose actual score is zero induces a ranking similar to that induced by
Katz's index).

In this paper, we provide convergence bounds in $\ell_\infty$ norm for SOR iterations for
the problem $(s  I-A)\bm x = \bm b$, where $A$ is a nonnegative real matrix
and $s  > \rho(A)$, in conditions of mild asynchrony, without any
additional hypothesis on $A$. Our main result are Theorem~\ref{teo:conv}, which shows 
that given a $\sigma <s$ and a vector $\w>0$ such that $A\w\leq\sigma\w$
SOR iterations reduce geometrically the $\w$-norm of the error (with a
computable contraction factor), and Theorem~\ref{teo:suitable}, which shows how
to compute such a $\w$ using only iterated products of $A$ with a vector. The
two results can be viewed as a constructive and computable version of the 
standard convergence results on SOR iteration based on the spectral radius.

We remark that SOR is actually not useful for PageRank, as shown recently by
Greif and Kurokawa~\cite{GrKNCSPP}. The author has found experimentally that the
same phenomenon plagues the computation of Katz's index. However, since
generalizing from \Gauss--Seidel to SOR does not bring any significant increase
in complexity in the proof, we decided to prove our results in the more general
setting.

\section{Step-asynchronous SOR}
\label{sec:SOR}

We now define \emph{step-asynchronous} SOR for the problem $(s  I-A)\x=\bm b$.
In general, \emph{asynchronous} SOR computes new values using arbitrarily old values; in this case, the hypotheses for convergence
are definitely stronger. In the 
\emph{partially asynchronous} case, instead, there is a 
finite limit on the ``oldness'' of the values used to compute new values, and while there is a decrease
in convergence speed, the hypotheses for convergence are essentially the same of
the sequential case~(see \cite{BeTPDCNM} for more details).

Step-asynchronous SOR uses the strictest possible time bound: one step. We thus perform a SOR-like update in arbitrary order:
\begin{equation}
\label{eq:gs}
x_i^{(t+1)} = (1-\omega)x_i^{(t)}+
\frac\omega{s -a_{ii}}\Bigl(b_i+\sum_{j\in N_i^{(t)}}a_{ij}x_j^{(t+1)}+
\sum_{j\in P_i^{(t)}\setminus\{\,i\,\}}a_{ij}x_j^{(t)}\Bigr). 
\end{equation}
The only constraint 
is that for each iteration an \emph{update total preorder}\footnote{A \emph{total preorder} is a 
set endowed with a reflexive and transitive total relation. We remark that a choice of a sequence
of such preorders is equivalent to a \emph{scenario} in the terminology of~\cite{BeTPDCNM}.} $\preceq^{(t)}$ of
the indices is given: $i\preceq^{(t)} j$ iff $x_i$ is updated before (or at the
same time of) $x_j$ at iteration $t$, and the set
$P_i^{(t)}$ of the indices for which we use the \emph{previous}
values is such that for all $j\succeq^{(t)} i$ we have $j\in P_i^{(t)}$, whereas
$N_i^{(t)}= n\setminus P_i{(t)}$ is the set indices for which we use the \emph{next}
values. Essentially,
we \emph{must} use previous values for all variables that are updated at the
same time of $x_i$ or after $x_i$, but we make no assumption on the remaining
variables. In this way we take into account cache incoherence, unpredictable
scheduling of multiple threads, and so on.\footnote{For example, if we have
exactly $n$ parallel updates at the same time we would have, in fact, a Jacobi
iteration: in that case, $N_i^{(t)}=\emptyset$ for all $i$.}

Matrixwise, the set $N_i^{(t)}$ induces a nonnegative
matrix $L^{(t)}$ given by
\[
L^{(t)}_{ij} = \Bigl[j\in N^{(t)}_i\Bigr]a_{ij}
\]
and a \emph{regular splitting}
\[
s  I-  A= \bigl(D- L^{(t)} \bigr) - R^{(t)},
\]
where $D=sI- \Diag A$ and $R^{(t)}$ is nonnegative with zeros on the diagonal.
Then, equation~(\ref{eq:gs}) can be rewritten as
\[
\bigl(D- \omega L^{(t)}\bigr)\x^{(t+1)} = (1-\omega)D\x^{(t)} +
\omega\bigl( \bm b +R^{(t)}\x^{(t)} \bigr).
\]
There is of course a permutation of row and columns (depending on $t$) such that 
$L^{(t)}$ is strictly lower triangular, but the only claim that can be made
about $R^{(t)}$ is that its diagonal is zero: actually, we could have
$L^{(t)}=0$ and $R^{(t)}=s  I - A - D$.

In particular, independently from the choice of $L^{(t)}$, if $\bar\x$ is a
solution we have as usual
\[
\bigl(D- \omega L^{(t)}\bigr)\bar\x = (1-\omega)D\bar\x +
\omega\bigl(\bm b  + R^{(t)}\bar\x \bigr)
\]
and
\begin{equation}
\label{eq:err}
\bigl(D- \omega L^{(t)}\bigr)\bigl(\bar\x-\x^{(t+1)}\bigr) =
(1-\omega)D\bigl(\bar\x-\x^{(t)}\bigr) + \omega
R^{(t)}\bigl(\bar\x-\x^{(t)}\bigr).
\end{equation}

\section{Suitability and convergence in $\w$-norm}
We now define suitability of a vector for a matrix, which will be the main tool in
proving our results. The idea is implicitly or explicitly at the core of several
classical proofs of convergence, and is closely related to that of \emph{generalized diagonal dominance}:
\begin{definition}
A vector $\w>0$ is \emph{$\sigma$-suitable} for $A$ if
$A\w\leq\sigma\w$.
\end{definition}

The usefulness of suitable vectors is that they induce norms
norms in which the decrease of the error caused by a SOR
iteration for of the problem $(sI-A) \x = \bm b$ can be controlled if
$s>\sigma$. If $A$ is irreducible, for instance, the dominant eigenvector is suitable
for the spectral radius, but it is exactly this kind of hypotheses that we want to avoid.
\begin{definition}
Given a vector $\w>0$, the $\w$-norm is defined by
\[\|\x\|^\w_\infty = \max_i\frac{|x_i|}{w_i}.\]
\end{definition}
The notation $\|\cdot\|_\infty^\w$ is used also for the operator norm induced
in the usual way.
We note a few useful properties---many others can be found in~\cite{BeTPDCNM}:
\begin{proposition}
\label{prop:wnorm}
Given a vector $\w$ that is $\sigma$-suitable for a nonnegative matrix $A$, the
following statements are true for all vectors $\x$:
\begin{enumerate}
\item\label{en:coord} $\bigl|x_i\bigr|\leq w_i\|\x\|_\infty^\w$;
\item\label{en:min} $\min_i w_i\|\x\|_\infty^\w\leq\|\x\|_\infty$;
\item\label{en:max} $\max_i w_i\|\x\|_\infty^\w\geq \|\x\|_\infty$;
\item\label{en:unit} $\|\w\|_\infty^\w=1$;
\item\label{en:Aw} $\|A\|_\infty^\w=\|A\w\|_\infty^\w$;
\item\label{en:mindef} if $\bm x\geq0$, $\|\bm
x\|^\w_\infty=\min\{\,\alpha\geq0\mid \x \leq \alpha \w\,\}$.
\item\label{en:suitable} $\|A\x\|_\infty^\w\leq\sigma\|\x\|_\infty^\w$; in
particular, $\rho(A)\leq\|A\|_\infty^\w\leq\sigma$.
\end{enumerate}
\end{proposition}
\begin{proof}
The first claims are immediate from the definition of $\w$-norm. For the last
claim,
\[
\|A\x\|_\infty^\w =\max_i\left|\frac{\sum_j a_{ij}x_j}{w_i}\right|\leq
\max_i\frac{\sum_j a_{ij}|x_j|}{w_i}=
\max_i\frac{\sum_j
a_{ij}w_j\|\x\|_\infty^\w}{w_i}\leq\sigma\|\x\|_\infty^\w.
\]
\end{proof}

The next theorem is based on the standard proof by induction of convergence
for SOR, but we make induction on the update time of a component rather than on its index,
and we use suitability to provide bounds to the norm of the error.

\begin{theorem}
\label{teo:conv}
Let $A$ be a nonnegative matrix and let $\w$
be $\sigma$-suitable for $A$.
Then, given $s>\sigma$ step-asynchronous SOR for the problem $(s  I- A)\x=\bm b$
converges for 
\[
% 0<\omega< \frac{s-\min_k a_{kk}}{s-\sigma} % >1, always
0<\omega< \frac2{\displaystyle 1+\max_k\frac{\sigma -a_{kk}}{s-a_{kk}}} %
% >1, always
\]
 and letting $\bar\x=(sI-A)^{-1}\bm b$ we have
\[
\big\|\bar\x - \x^{(t+1)}\bigr\|_\infty^\w
\leq r\bigl\|\bar\x -
\x^{(t)}\bigr\|_\infty^\w,
\]
where 
\[
r=|1-\omega| +\omega\max_k\frac{\sigma -a_{kk}}{s-a_{kk}}<1.
\]
Moreover, 
\[
\bigl\|\bar\x - \x^{(t+1)}\bigr\|_\infty^\w
\leq\frac r{1-r}
\bigl\|\x^{(t+1)} - \x^{(t)}\bigr\|_\infty^\w.
\]
\end{theorem}
\begin{proof}
Let $\preceq^{(t)}$ be a sequence of update orders, and 
$P_i{(t)}$ a sequence of previous-value sets, one for each step $t$ and
variable index $i$, compatible with the respective update orders.
We work by induction on the order
$\preceq^{(t)}$, proving the statement
\begin{equation}
\label{eq:ind}
\bigl| e_i^{(t+1)}\bigr|\leq\biggl( 
|1-\omega|+
\omega\frac{\sigma-a_{ii}}{s-a_{ii}}\biggr)w_i\bigl\|\bm
e^{(t)}\bigr\|_\infty^\w,
\end{equation}
where $\bm e^{(t)}=\bar\x - \x^{(t)}$, assuming it is true for all
$k\prec^{(t)}i$.

Note that for all $i$
\[
0<\frac{\sigma-a_{ii}}{s-a_{ii}}<1,
\]
so for $0<\omega\leq 1$
\[
|1-\omega|+
\omega\frac{\sigma-a_{ii}}{s-a_{ii}} = 1 -
\omega\biggl(1-\frac{\sigma-a_{ii}}{s-a_{ii}}\biggr)<1,
\]
and analogously for
\[
1\leq\omega<\frac2{\displaystyle 1+ \max_k\frac{\sigma -a_{kk}}{s-a_{kk}}} %
% >1,
\]
we have
\[
|1-\omega|+
\omega\frac{\sigma-a_{ii}}{s-a_{ii}} = 
\omega\biggl(1+\frac{\sigma-a_{ii}}{s-a_{ii}}\biggr)-1<
\frac2{\displaystyle 1+\max_k\frac{\sigma
-a_{kk}}{s-a_{kk}}}\biggl(1+\frac{\sigma-a_{ii}}{s-a_{ii}}\biggr)-1<1.
\]

Writing explicitly~(\ref{eq:err}) for the $i$-th coordinate, we have
\begin{align*}
\bigl|e_i^{(t+1)}\bigr| &= \biggl|(1-\omega)e_i^{(t)}+
\frac\omega{s -a_{ii}}\Bigl(\sum_{j\in N_i^{(t)}} a_{ij}e_j^{(t+1)}
+\sum_{j\in P_i^{(t)}\setminus\{\,i\,\}} a_{ij}e_j^{(t)}\Bigr) 
\biggr|.
\end{align*}
Since $j\in N_i^{(t)}$ implies by definition $j\prec^{(t)}i$, we can apply the
induction hypothesis on $e_j^{(t+1)}$ to state that $e_j^{(t+1)}\leq w_j\bigl\|\bm e^{(t)}\bigr\|_\infty^\w$.
The same bound applies to $e_j^{(t)}$ using
the first statement of Proposition~\ref{prop:wnorm}. 

We now notice that $\sigma$-suitability implies 
\[(A-\Diag A)\w \leq (\sigma I-\Diag A)\w,\]
which in coordinates tells us that
\[\sum_{j\neq i} a_{ij} w_j \leq \bigl(\sigma-a_{ii})w_i.\]
Thus,
\begin{align*}
\bigl|e_i^{(t+1)}\bigr| &\leq 
\biggl(|1-\omega|w_i +
\omega\frac1{s-a_{ii}} \Bigl(\sum_{j\in N_i^{(t)}}
a_{ij}w_j +\sum_{j\in P_i^{(t)}\setminus\{\,i\,\}} a_{ij}w_j\Bigr) 
\biggr)\bigl\|\bm e^{(t)}\bigr\|_\infty^\w\\
&\leq\biggl(|1-\omega|+
\omega\frac{\sigma-a_{ii}}{s-a_{ii}} 
\biggr)w_i\bigl\|\bm e^{(t)}\bigr\|_\infty^\w.
% \\&\leq\biggl( 1 -
% \omega\frac{s-\sigma}{s-\min_k a_{kk}}\biggr)w_i\bigl\|\bm
% e^{(t)}\bigr\|_\infty^\w.\\
\end{align*}
By the very definition of $\w$-norm,~(\ref{eq:ind}) yields
\[
\bigl\| \bm e_i^{(t+1)}\bigr\|_\infty^\w\leq\biggl( 
|1-\omega|+
\omega\max_k\frac{\sigma-a_{kk}}{s-a_{kk}}\biggr)\bigl\|\bm
e^{(t)}\bigr\|_\infty^\w.
\]

For the second statement, we have
\begin{multline*}
 \bigl\| \x- \x^{(t)}\bigr\|_\infty^\w - \bigl\| \x^{(t+1)} -
 \x^{(t)}\bigr\|_\infty^\w   \leq \bigl\| \x - \x^{(t+1)} + \x^{(t)} - \x^{(t)}\bigr\|_\infty^\w\\= \bigl\|\x - \x^{(t+1)}\bigr\|_\infty^\w\leq r \bigl\|\x - \x^{(t)}\bigr\|_\infty^\w,
\end{multline*}
whence 
\[
\bigl\|\x - \x^{(t+1)}\bigr\|_\infty^\w \leq r \bigl\|\x -
\x^{(t)}\bigr\|_\infty^\w \leq\frac r{1-r} \bigl\|\x^{(t+1)} -
\x^{(t)}\bigr\|_\infty^\w.
\]
\end{proof}

We remark that the smallest contraction factor is obtained when $\omega=1$, that is, with no relaxation. This does not mean,
however, that relaxation is not useful: convergence might be faster with $\omega\neq 1$; it is just that the error bound we
provide features the best constant when $\omega=1$.

\begin{corollary}
\label{cor:bound}
With the same hypotheses and notation of Theorem~\ref{teo:conv}, step-asynchronous \Gauss--Seidel iterations converge and
\[
\bigl\|\bar\x - \x^{(t+1)}\bigr\|_\infty^\w
\leq\frac{\displaystyle \max_k\frac{\sigma
-a_{kk}}{s-a_{kk}}}{\displaystyle 1-\max_k\frac{\sigma
-a_{kk}}{s-a_{kk}}}
\bigl\|\x^{(t+1)} - \x^{(t)}\bigr\|_\infty^\w.
\]
\end{corollary}

\begin{corollary}
\label{cor:bound}
Let $A$ be an irreducible nonnegative matrix and $\w$ its dominant eigenvector. Then the statement of Theorem~\ref{teo:conv} 
is true in $\w$-norm with $\sigma=\rho(A)$.
\end{corollary}

A simple consequence is that if we know a $\sigma$-suitable vector $\w$ for
$A$ we can just behave as if the step-asynchronous SOR is converging in the
standard supremum norm, but we have a reduction in the strength 
of the bound given by the ratio between the maximum
and the minimum component of $\w$:
\begin{corollary}
\label{cor:bound}
With the same hypotheses and notation of Theorem~\ref{teo:conv}, step-asynchronous \Gauss--Seidel iterations converge and
\[
\bigl\|\bar\x - \x^{(t+1)}\bigr\|_\infty
\leq \frac{\max_i w_i}{\min_i
w_i}\frac{\displaystyle \max_k\frac{\sigma
-a_{kk}}{s-a_{kk}}}{\displaystyle 1-\max_k\frac{\sigma
-a_{kk}}{s-a_{kk}}}
\bigl\|\x^{(t+1)} - \x^{(t)}\bigr\|_\infty.
\]
\end{corollary}
\begin{proof}
An application of Proposition~\ref{prop:wnorm}.\ref{en:min}
and~\ref{prop:wnorm}.\ref{en:max}.
\end{proof}

We remark that 
\[\max_k\frac{\sigma
-a_{kk}}{s-a_{kk}}\leq \frac\sigma s,\]
so it is possible to restate all results in a simplified (but less powerful)
form.

\section{Practical issues}
In principle it is always better to compute the actual $\w$-norm,
rather than using the rather crude bound of Corollary~\ref{cor:bound}.\footnote{The bound is actually \emph{very} crude, in particular on
reducible matrices when $\sigma$ is close to $\rho(A)$.} On the other
hand, computing the $\w$-norm requires storing and accessing $\w$, which could be expensive. 

In practice, it is convenient to restrict oneself to vectors $\w$ satisfying $\|\w\|_\infty=1$, as in that case
$\|\x\|_\infty\leq\|\x\|_\infty^\w$, and for some $\x$ we actually have equality.
Then, we can store in few bits an approximate vector $\w'\leq \w$, which can be used 
to estimate $\bigl\|\bar\x^{(t+1)} - \x^{(t)}\bigr\|_\infty^\w$, as we have, using the notation of Theorem~\ref{teo:conv},
\[
\bigl\|\bar\x - \x^{(t+1)}\bigr\|_\infty\leq \bigl\|\bar\x - \x^{(t+1)}\bigr\|_\infty^\w\leq \frac r{1-r} \bigl\|\x^{(t+1)}
- \x^{(t)}\bigr\|_\infty^\w
\leq 
\frac r{1-r} \bigl\|\x^{(t+1)} -
\x^{(t)}\bigr\|_\infty^{\w'}.
\]
A reasonable choice is that of keeping in memory $\lceil -\log_2
w_i\rceil$. Using a byte of storage we can keep track of $w_i$'s no smaller than
Moreover, during the evaluation of the norm we just have to multiply by a power
of two, which can be done very quickly in IEEE 754 format.

\section{Choosing a suitable vector}
\label{sec:choosing}

We now come to the main result: given a nonnegative matrix $A$ 
and a $\sigma>\rho(A)$, it is possible (constructively) to compute a vector
$\w$ that is $\sigma$-suitable for $A$. In essence, the
computation of a $\sigma$-suitable vector for $A$ ``tames'' the
non-normality of the iterative process, at the price of a reduction of the
convergence range.

\begin{theorem}
\label{teo:suitable}
Let $A$ be nonnegative and $\sigma>\rho(A)$. Let
\[
\w_\sigma^{(k)} = \sum_{i=0}^k \biggl(\frac A{\sigma}\biggr)^i\1
\]
and
\[
\w_\sigma =\lim_{k\to\infty}\w_\sigma^{(k)}
\]
Then, $A\w_\sigma<\sigma\w_\sigma$. In particular, $\w_\sigma$ is
$\sigma$-suitable for $A$, and there is a $k$ such that
\[
A\w_\sigma^{(k)}\leq
\sigma\w_\sigma^{(k)},
\]
so $\w_\sigma^{(k)}$ is
$\sigma$-suitable for $A$.
\end{theorem}
\begin{proof}
Consider the matrix $A+\delta\1\1^*$, where $\delta>0$. Since it is strictly
positive, the Perron--Frobenius Theorem tells us that there is a 
dominant eigenvector $\w_\delta>0$. Moreover, since for $\delta\to\infty$
we have $\rho(A+\delta\1\1^*)\to\infty$, and the spectral radius is continuous in the
matrix entries, there must be a $\delta_\sigma$ such that
\[
\rho\bigl(A+\delta_\sigma\1\1^*\bigr)=\sigma.
\]
We have
\begin{align*}
\bigl(A+\delta_\sigma\1\1^*\bigr)\w_{\delta_\sigma}&=\sigma\w_{\delta_\sigma}\\
A\w_{\delta_\sigma}+\delta_\sigma\|\w_{\delta_\sigma}\|_1\1
&=\sigma\w_{\delta_\sigma}\\
\frac{\delta_\sigma\|\w_{\delta_\sigma}\|_1}{\sigma}\1
&=\biggl(1-\frac A{\sigma}\biggr)\w_{\delta_\sigma}\\
\w_{\delta_\sigma}&=\frac{\delta_\sigma\|\w_{\delta_\sigma}\|_1}{\sigma}\sum_{i=0}^\infty\biggl(\frac
A{\sigma}\biggr)^i \1.
\end{align*}
We now observe that the scaling factor is irrelevant: $\w_{\delta_\sigma}$ is an
eigenvector, so it is defined up to a multiplicative constant. We can thus just write
\[
\w_\sigma =\sum_{i=0}^\infty\biggl(\frac
A{\sigma}\biggr)^i \1
\]
and state that
\[
\bigl(A+\delta_\sigma\1\1^*\bigr)\w_\sigma =\sigma\w_\sigma,
\]
which implies
\[
A\w_\sigma =\sigma\w_\sigma  - \delta_\sigma\|\w_\sigma\|_1\1<\sigma\w_\sigma.
\]
Thus, as $\w_\sigma^{(k)}\to \w_\sigma$ when $k\to \infty$, for some
$k$ we must have
\[
A\w_\sigma ^{(k)}\leq\sigma\w_\sigma ^{(k)}.
\]
\end{proof}

The previous theorem suggests the following procedure. Under the given
hypotheses, start with $\w^{(0)}=\1$, and iterate
\begin{align*}
\bm z&=A \w^{(t)}\\
\w^{(t+1)}&=\bm z/\sigma+\1.
\end{align*}
Note that this is just a Jacobi iteration for the problem $(I-A/\sigma)\bm
x=\1$, which is natural, as $\w_\sigma$ is just its solution.
The iteration stops as soon as
\begin{equation}
\label{eq:norm}
\max_i \frac{z_i}{w^{(t)}_i}\leq \sigma,
\end{equation}
and at that point $\w^{(t)}$ is by definition $\sigma$-suitable for $A$, 
so we can apply Theorem~\ref{teo:conv}. 

In practice, it is useful to keep the current vector
$\w^{(t)}$ normalized: just set $s^{(0)}=1$ at the start, and then iterate
\begin{align*}
\bm z&=A \w^{(t)}\\
\bm u&=\bm z/\sigma+s^{(t)}\1\\
s^{(t+1)} &= s^{(t)} / \|\bm u\|_\infty\\
\w^{(t+1)}&=\bm u / \|\bm u\|_\infty.
\end{align*}

% Essentially, we are ``taming'' the non-normality of $A$ at the price of a
% reduction in the range of convergence (i.e., $s >\sigma$ instead of
% $s >\rho(A)$).
We remark that, albeit used for clarity in the statement of Theorem~\ref{teo:suitable}, the (exact) knowledge of $\rho(A)$ is
not strictly necessary to apply the technique above: indeed, if the procedure terminates $\sigma\geq\rho(A)$ by
Proposition~\ref{prop:wnorm}.  

There are a few useful observations about the behavior of the normalized
version of the procedure. First, if $\sigma < \rho(A)$ necessarily $s^{(t)}\to 0$ as
$t\to \infty$. Second, by Collatz's classical bound~\cite{ColECZM},
the maximum in~(\ref{eq:norm}) is an upper bound to $\rho(A)$. This happens
without additional hypotheses\footnote{We report the following two easy proofs
as in most of the literature Collatz's bounds are proved for irreducible
matrices using Perron--Frobenius theory.} on $A$ because whenever $A\x\leq
\gamma\x$ with $\x>0$ we have
\[
\rho(A)\leq \|A\|^\x_\infty=\|A\x\|^\x_\infty\leq\|\gamma \x\|^\x_\infty=\gamma.
\]
If, moreover, we compute also the minimum ratio
\begin{equation}
\label{eq:min}
\min_i \frac{z_i}{w^{(t)}_i},
\end{equation}
this is a lower bound to $\rho(A)$, again without additional hypotheses on $A$.
Indeed, note that whenever $\beta\x\leq A\x$ with
$\x\geq0$, for every $\delta>0$ if $\w$ is a positive eigenvector of
$A+\delta\1\1^*$ we have
\[
\beta\x\leq A\x\leq (A+\delta\1\1^*)\x\leq
(A+\delta\1\1^*)\|\x\|^\w_\infty \w = \rho(A+\delta\1\1^*)\|\x\|^\w_\infty \w.
\]
The last inequality implies $\beta\leq\rho(A+\delta\1\1^*)$ by
Proposition~\ref{prop:wnorm}.\ref{en:mindef}, and since the inequality is
true for every $\delta$ it is true by continuity also for $\delta=0$.

These properties suggest that in practice iteration should be stopped if
$s^{(t)}$ goes below the minimum representable floating-point number: in this case, either
$\sigma<\rho(A)$, or the finite precision at our disposal is not sufficient to
compute a suitable vector because we cannot represent correctly a transient
behavior of the powers of $A$.

If instead the minimum~(\ref{eq:min}) becomes larger than $\sigma$, we can
safely stop: unfortunately, the latter event cannot be guaranteed to happen when
$\sigma<\rho(A)$ without additional hypotheses on $A$ (e.g., irreducibility):
for instance, if $A$ has a null row the minimum~(\ref{eq:min}) will always be equal to zero.

% Second, 
% the maximum actually \emph{converges} to $\rho(A)$.
% As a consequence, for every $\delta>0$ there is at $t$ such that
% $s^{(t)}\leq\delta$. Consider a subsequence $t_i$ such that $s^{(t_i)}\leq
% 1/i$. For each $i$, the remaining computation is now equivalent to a power
% method iteration on the positive matrix $A/\sigma+c\1\1^*/i$ for come constant
% $c$ independent of $i$. Since $\rho(A/\sigma+c\1\1^*/i)\to\rho(A/\sigma)$

Of course, there ain't no such thing as a free lunch. The termination
of the process above is guaranteed if $\sigma>\rho(A)$, but we have no indication of how many step will be
required. Moreover, in principle some of the coordinates of the suitable vector could be so small 
to make Theorem~\ref{teo:conv}
unusable. For $\sigma$ close to $\rho(A)$
convergence can be very slow, as it is related to the convergence of Collatz's
lower and upper bounds for the dominant eigenvalue.

Nonetheless, albeit all of the above must happen in pathological cases, we show on
a few examples that, actually, in real-world cases computing a $\sigma$-suitable
vector is not difficult.

We remark that in principle any dyadic product $\bm u\bm v^*$ such that $A+\bm u
\bm v^*$ is irreducible will do the job in the proof of
Theorem~\ref{teo:suitable}. There might be choices (possibly depending on
$A$) for which the computation above terminates more quickly.

\section{Examples}

\subsection{Bounding the error of $(I-A)\bm x=\bm b$}

If $A$ is nonnegative matrix with $\rho(A)<1$, then $I-A$ is invertible and the 
problem $(I-A)\bm x=\bm b$ has a unique solution, and in the limit we have
convergence geometric in $\rho(A)$. However, if we choose a
$1>\sigma>\rho(A)$ (say, $\sigma=(1+\rho(A))/2$) and a $\sigma$-suitable vector $\w$, the bounds of
Theorem~\ref{teo:conv} will be valid, so we will be
able to control the error in $\w$-norm.

\subsection{Katz's index}

Let $M$ be a nonnegative matrix (in the standard formulation, the adjacency
matrix of a graph). Then, given $\alpha<1/\rho(M)$ Katz's index is defined by
\[
\bm k^*= \bm v^*\bigl(1-\alpha M\bigr)^{-1}= \bm v^*\sum_{k\geq
0}\alpha^kM^k,
\]
where $\bm v$ is a \emph{preference vector}, which is just $\mathbf 1$ in Katz's
original definition~\cite{KatNSIDSA}.\footnote{We must note that actually Katz's index is 
$\bm v^*\bigl(1-\alpha M\bigr)^{-1}M$. This additional multiplication by
$M$ is somewhat common in the literature; it is probably a case of
\textit{horror vacui}.}.

If we want to apply Theorem~\ref{teo:conv}, we must choose a $\sigma>\rho(A)$
and a $\sigma$-suitable vector $\w$ for $A$. The vector can then be used to
accurately estimate the computation of Katz's index for all
$\alpha<1/\sigma$. This property is particularly useful, as it is common to
estimate the index for different values of $\alpha$, and to that purpose it is
sufficient to compute once for all a $\sigma$-suitable vector for  
a $\sigma$ chosen sufficiently close to $\rho(A)$.

\subsection{PageRank} The case of PageRank is similar to Katz's index. We have
\[ \bm r^*= (1-\alpha)\bm v^*(1-\alpha P)^{-1}= (1-\alpha)\bm
v^*\sum_{k=0}^\infty\alpha^kP^k, \] where $\bm v$ is the preference vector, and
$P=\bar G+\bm d\bm u^*$ is a stochastic matrix; $\bar G$ is the adjacency matrix
of a graph $G$, normalized so that each nonnull row adds to one, $\bm d$ is the
characteristic vector of \emph{dangling nodes} (nodes without outlinks, i.e.,
null rows), and $\bm u$ is the dangling-node distribution, used to redistribute
the rank lost through dangling nodes. It is common to use a uniform $\bm u$, but
most often $\bm u=\bm v$, and in that case we speak of \emph{strongly
preferential} PageRank~\cite{BSVPFD}.

We remark that in the latter case it is well known that the \emph{pseudorank}
\[
\bm p^*= (1-\alpha)\bm v^*\sum_{k=0}^\infty\alpha^k\bar G^k 
\]
satisfies
\[
\bm r = \frac{\bm p}{\|\bm p\|_1}.
\]
That is, PageRank and the pseudorank are parallel vectors. This is relevant for
the computation of several strongly preferential PageRank vectors: just
compute a $\sigma$-suitable vector for $\bar G$ (rather than
one for each $\bar G+\bm d\bm v^*$, depending on $\bm v$), and compute
pseudoranks instead of ranks.

The case of PageRank is however less interesting because, as David Gleich made
the author note, assuming the notation of Section~\ref{sec:SOR} and $\omega=1$
\begin{align*}
 \bigl(1-\alpha P^T\bigr)\bm x^{(t+1)} -(1-\alpha)\bm v &= \bigl(D- L^{(t)}-
 R^{(t)}\bigr)\bm x^{(t+1)}  -(1-\alpha)\bm v\\
&=  \bigl(D- L^{(t)}\bigr)\bm x^{(t+1)} - R^{(t)}\bm x^{(t+1)}  -(1-\alpha)\bm
 v\\
&=  R^{(t)}\bm x^{(t)}+ (1-\alpha)\bm
 v - R^{(t)}\bm x^{(t+1)}  -(1-\alpha)\bm
 v\\
 &= R^{(t)}(\bm x^{(t)}-\bm x^{(t+1)}).
\end{align*}
Since $\bigl\|R^{(t)}\bigr\|_1\leq \alpha$, we can $\ell_1$-bound the residual
\[
\Bigl\|\bigl(1-\alpha P^T\bigr)^{-1}\Bigr\|_1 =
\Bigl\|\sum_{k=0}^\infty\alpha^k\bigl(P^T\bigr)^k\Bigr\|_1\leq
\frac1{1-\alpha}
\]
we conclude that
\[
\bigl\|\bar{\bm x} -\bm x^{(t+1)}\bigr\|_1\leq\frac\alpha{1-\alpha}\bigl\|\bm
x^{(t+1)}-\bm x^{(t)}\bigr\|_1.
\]
It is thus possible, albeit wasteful, to bound the supremum norm of the error
using its $\ell_1$ norm.

\section{Experiments}

In this section we discuss some computational experiments involving the computation of PageRank
and Katz's index on real-world graphs. We focus on a snapshot of the English version of Wikipedia taken in 2013
(about four million nodes and one hundred million arcs) and a snapshot of the \texttt{.uk} web domain
taken in may 2007 (about one hundred million nodes and almost four billion arcs).\footnote{Both datasets are
publicly available at the site of the Laboratory for Web Algorithmics (\texttt{http://law.di.unimi.it/}) under the identifiers
\texttt{enwiki-2013} and \texttt{uk-2007-05}.} These two graphs have some structural differences, which we highlight
in Table~\ref{tab:datasets}.

\begin{table}
\centering
\begin{tabular}{l|rr}
& \multicolumn{1}{c}{Wikipedia} & \multicolumn{1}{c}{\texttt{.uk}}\\
\hline
nodes & $4\,206\,785$ & $105\,896\,555$ \\
arcs & $101\,355\,853$ & $3\,738\,733\,648$\\
avg. degree & $24.093$ & $35.306$\\
giant component & $89.00\%$ & $64.76\%$ \\
harmonic diameter & $5.24$ & $22.78$\\
dominant eigenvalue & $191.11$ & $5676.63$ \\  
\end{tabular}
\caption{\label{tab:datasets}Basic structural data about our two datasets.}
\end{table}

We applied the procedure described in Section~\ref{sec:choosing} to the system associated with PageRank
and Katz's index, with $\sigma\in\{\,1/(1-2^{-i})\mid 1\leq i\leq 10\,\}$ for PageRank and
$\sigma\in\{\,\lambda/(1-2^{-i})\mid 1\leq i\leq 10\,\}$ for Katz's index.

\begin{figure}
\centering
\includegraphics[scale=.6]{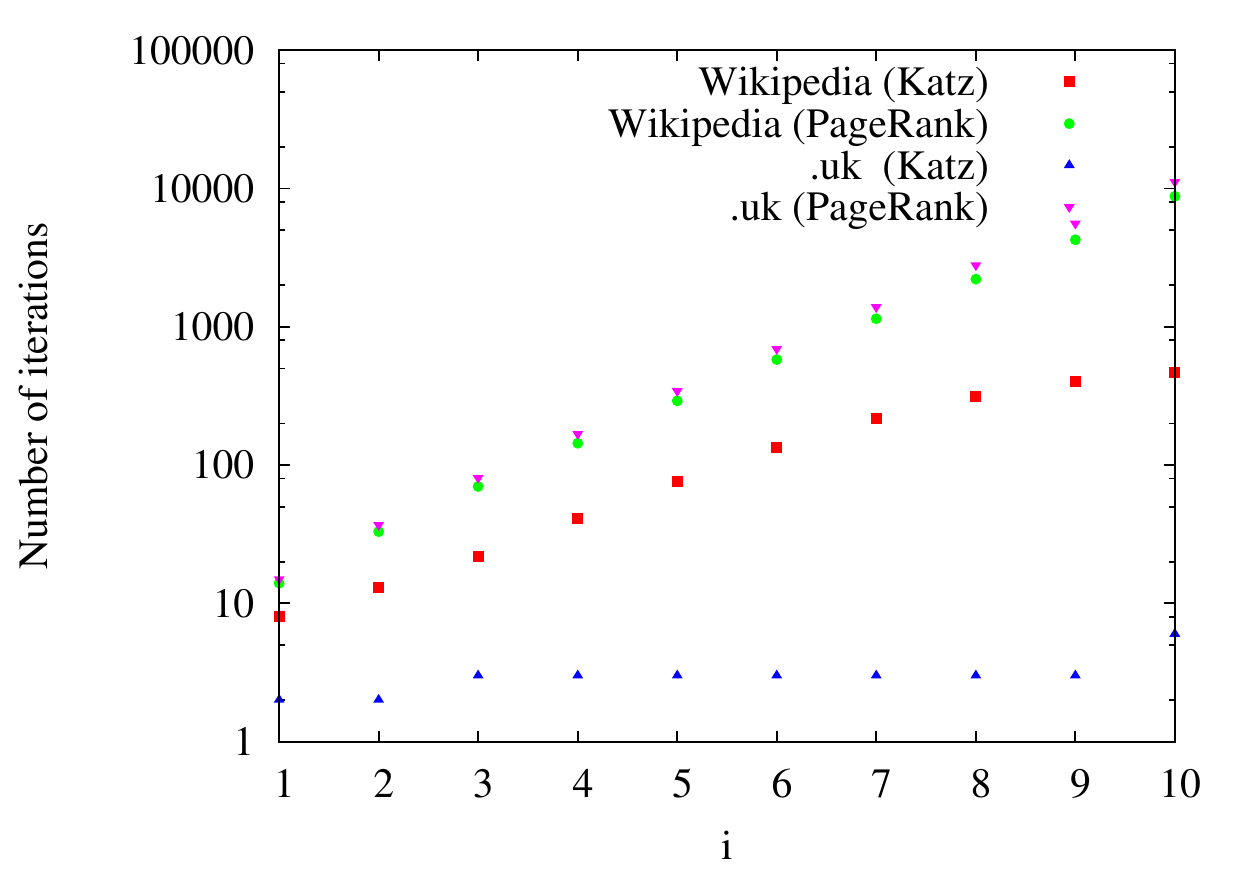}
\caption{\label{fig:iter}Number of iterations that are necessary to compute a $\lambda/(1-2^{-i})$-suitable vector.}
\end{figure}

In Figure~\ref{fig:iter} we report the number of iterations that are necessary to compute the $i$-th suitable
vector. The two datasets show the same behavior in the case of PageRank---an exponential
increase in the number of iterations as we get exponentially closer to the limit value. The case of Katz is more varied: whereas Wikipedia has a significant
growth in the number of iterations (but clearly slower than the PageRank case), \texttt{.uk} has a minimal variation 
across the range (from $2$ to $6$).  

\begin{figure}[htb]
\centering
\includegraphics[scale=.45]{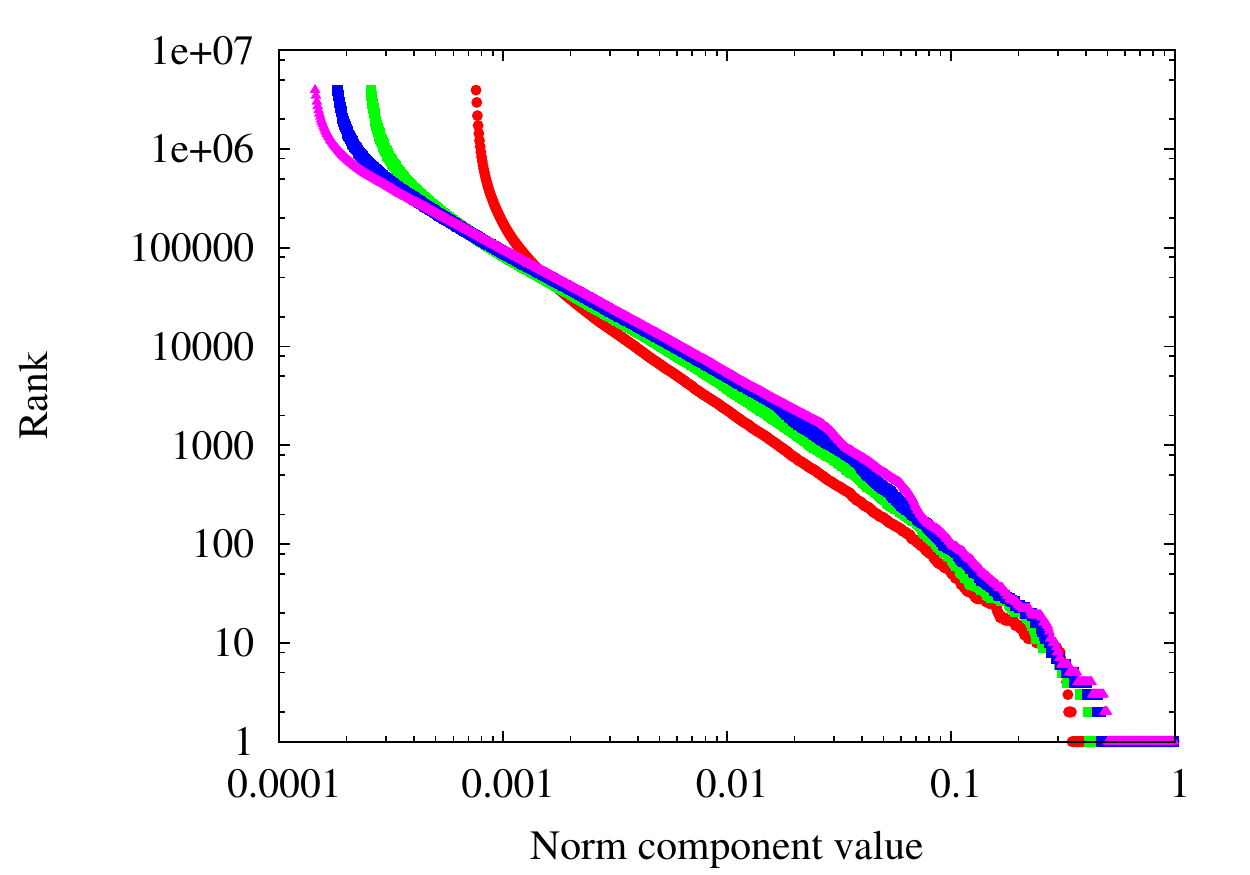}\quad\includegraphics[scale=.45]{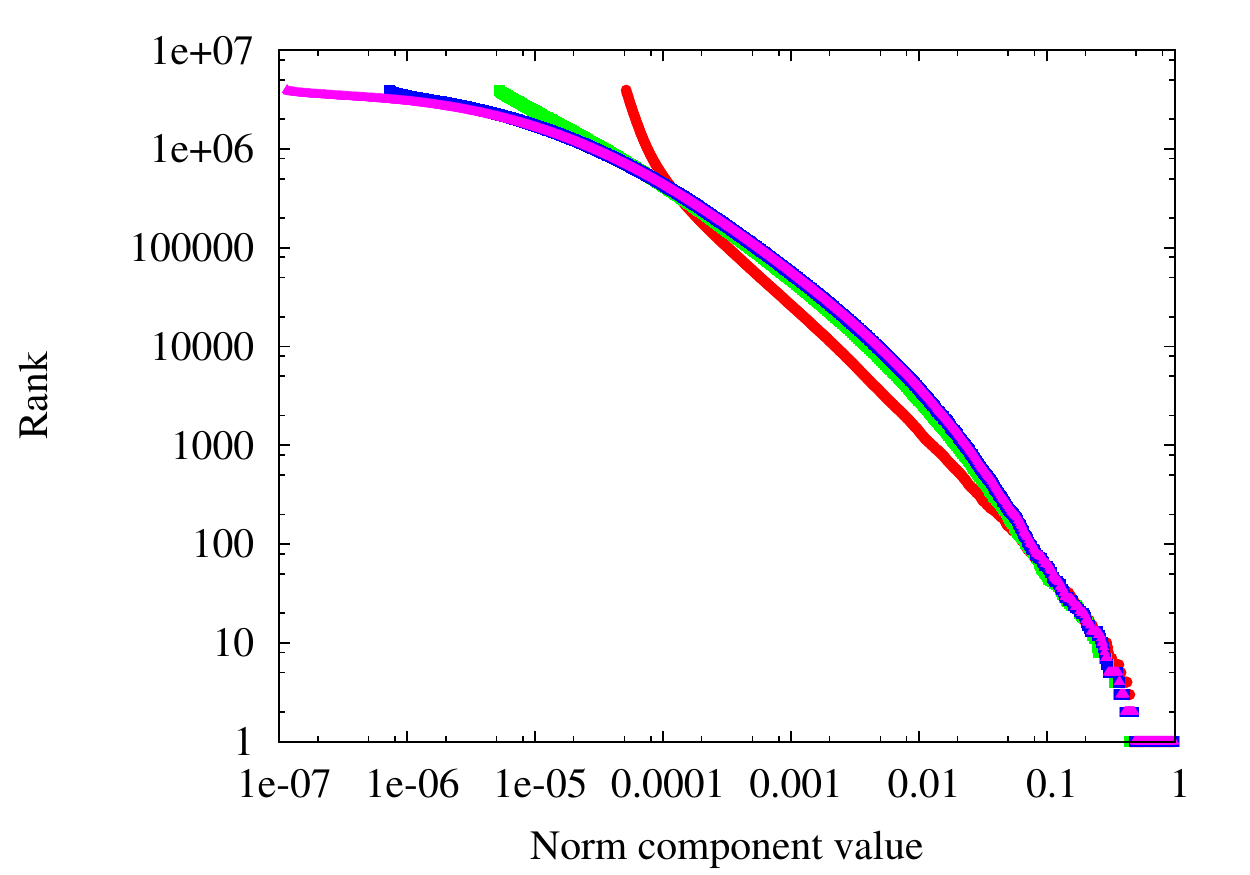}\\
\includegraphics[scale=.45]{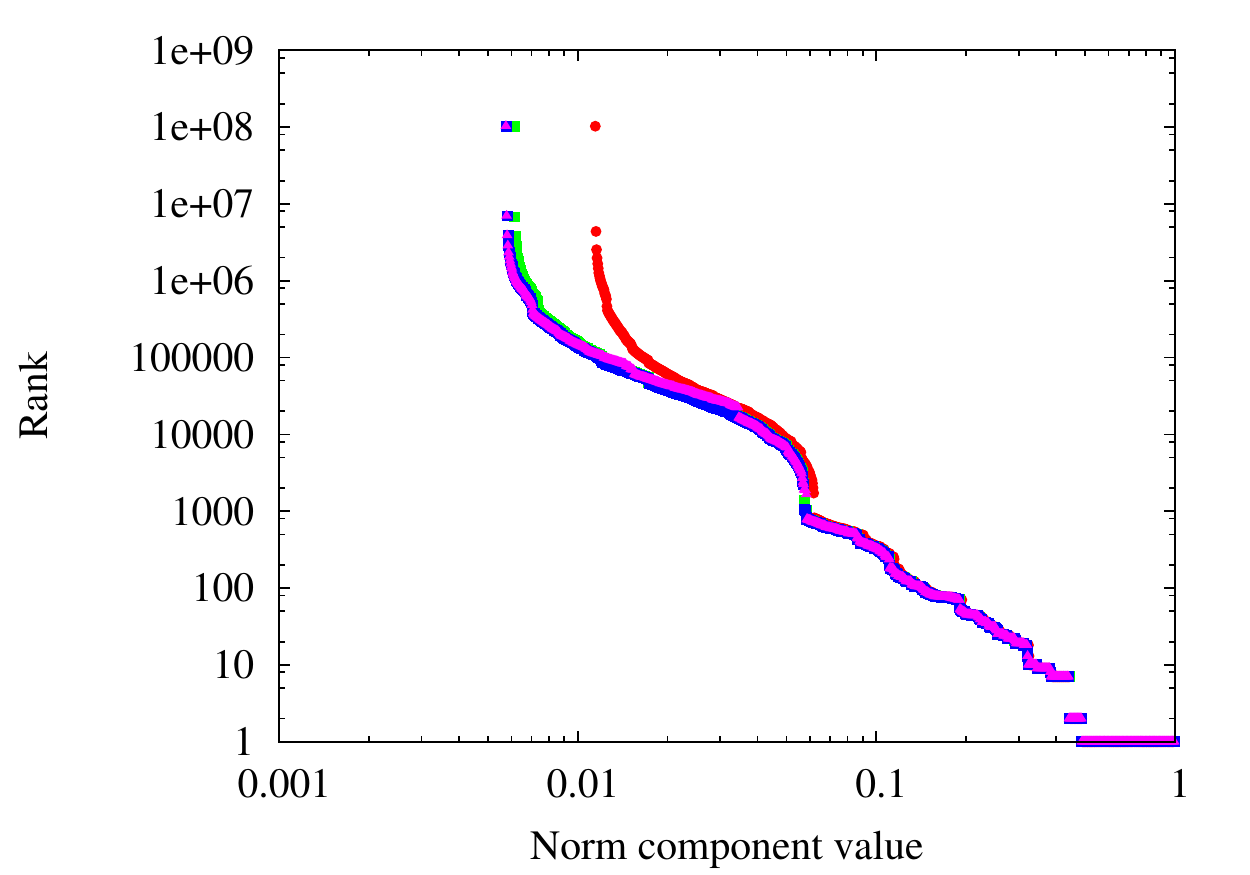}\quad\includegraphics[scale=.45]{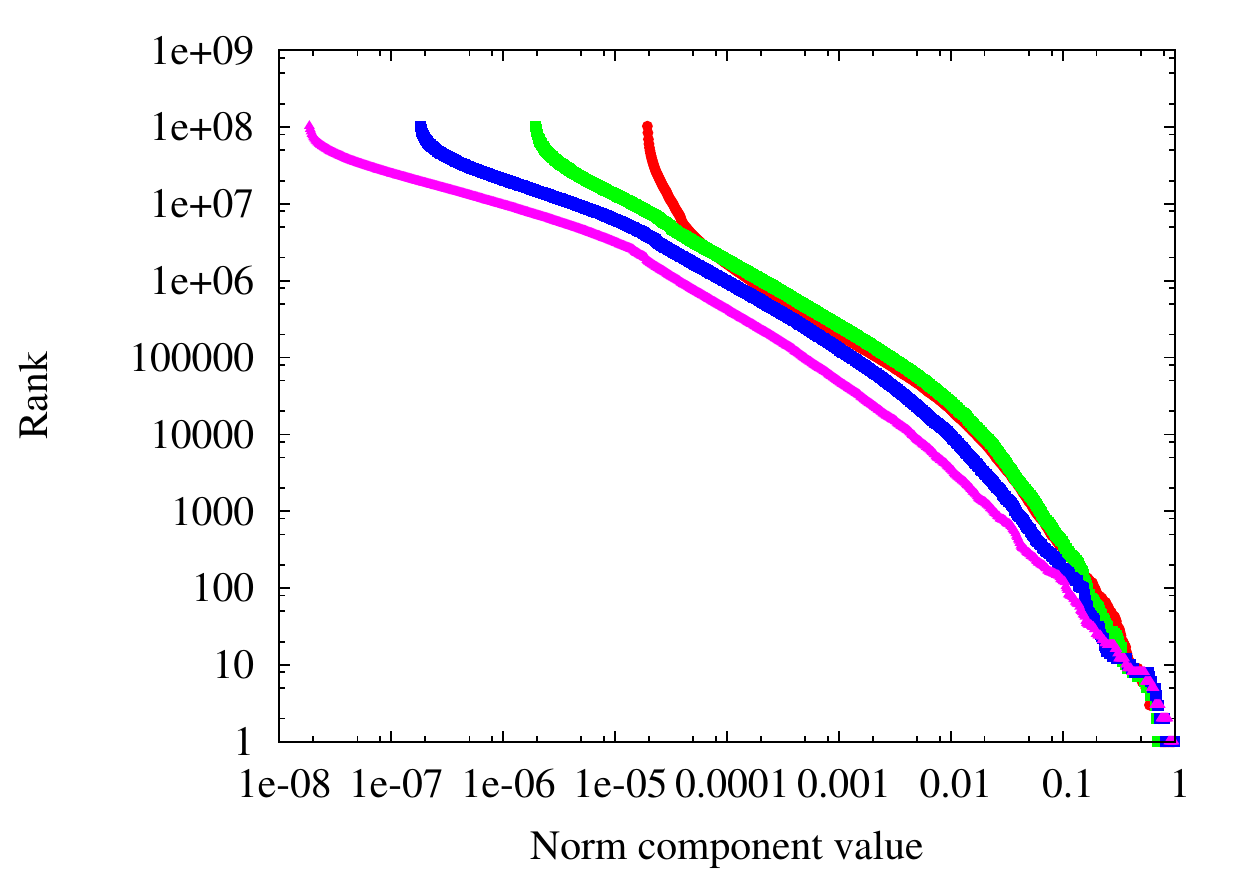}\\
\caption{\label{fig:norm}Exponentially binned frequency plots of the values of $\lambda/(1-2^{-i})$-suitable vectors, $i=1$, $4$, $7$ and $10$.}
\end{figure}

In Figure~\ref{fig:norm} we draw the (exponentially binned) distribution of values of suitable vectors for
a choice of four equispaced values of $i$. The vectors are normalized in $\ell_\infty$ norm, that is,
the largest value is one.

The shape of the distribution depends both on the graph
and on the type of centrality computed, but two features are constant: first, as we approach $\lambda$ the distribution
contains smaller and smaller values; second, the smallest value in the PageRank case is several orders of magnitude smaller.

Smaller values imply a larger $\w$-norm: indeed, one can think of the elements of an $\ell_\infty$-normalized suitable
vector $\w$ as weights that ``slow down'' the convergence of problematic nodes by inflating their raw error. The intuition
we gather from the distribution of values is that bounding the convergence of PageRank is more difficult. 

\section{Conclusions}

We have presented results that make it possible to bound the supremum norm of
the absolute error of SOR iterations an $M$-matrix $sI-A$ even when estimating
$\bigl\|(sI-A)^{-1}\bigr\|_\infty$ is not feasible. Rather than
 relying on additional hypotheses such as positive definiteness, irreducibility
 and so on, our results suggest to
compute first a $\sigma$-\emph{suitable} positive vector $\w$ with the property
that SOR iterations converge geometrically in $\w$-norm by a computable factor.

While we cannot bound without additional
hypotheses the resources (number of iterations and precision) that are necessary
to compute $\w$, in practice the computation is not difficult, and given an
$M$-matrix $sI-A$ the associated $\sigma$-suitable $\w$ can be used for all
$s>\sigma$.  

\hyphenation{ Vi-gna Sa-ba-di-ni Kath-ryn Ker-n-i-ghan Krom-mes Lar-ra-bee
  Pat-rick Port-able Post-Script Pren-tice Rich-ard Richt-er Ro-bert Sha-mos
  Spring-er The-o-dore Uz-ga-lis }

\end{document}